\documentclass[onecolumn,10pt]{article}
\usepackage[top=2.54cm, bottom=2.54cm, left=2.54cm, right=2.54cm,a4paper]{geometry}

\usepackage[utf8]{inputenc} 
\usepackage{textcomp} 

\usepackage{amsmath,amsthm}
\usepackage{bm}  
\usepackage{bbm}

  \DeclareFontFamily{U}{mathb}{\hyphenchar\font45}
\DeclareFontShape{U}{mathb}{m}{n}{
      <5> <6> <7> <8> <9> <10> gen * mathb
      <10.95> mathb10 <12> <14.4> <17.28> <20.74> <24.88> mathb12
      }{}
\DeclareSymbolFont{mathb}{U}{mathb}{m}{n}

\DeclareFontFamily{U}{matha}{\hyphenchar\font45}
\DeclareFontShape{U}{matha}{m}{n}{
      <5> <6> <7> <8> <9> <10> gen * matha
      <10.95> matha10 <12> <14.4> <17.28> <20.74> <24.88> matha12
      }{}
\DeclareSymbolFont{matha}{U}{matha}{m}{n}

\DeclareMathSymbol{\oasterisk}{3}{matha}{"66}
\DeclareMathSymbol{\boxasterisk}{3}{mathb}{"66}

\usepackage[backend=biber,sorting=none,style=phys,biblabel=brackets,backref=false,bibencoding=utf8,maxnames=20,eprint=true]{biblatex}
\makeatletter
\pretocmd{\blx@head@bibintoc}{\phantomsection}{}{\ddt}
\makeatother
\DefineBibliographyStrings{english}{%
  backrefpage = {page},
  backrefpages = {pages},
}
\usepackage{titlesec}
\usepackage{caption}
\titleformat*{\section}{\bfseries}
\titleformat*{\subsection}{\normalsize\bfseries}
\titleformat*{\subsubsection}{\bfseries}
\titleformat*{\paragraph}{\large\bfseries}
\titleformat*{\subparagraph}{\large\bfseries}

\titlespacing\section{0pt}{12pt plus 4pt minus 2pt}{2pt plus 2pt minus 2pt}

\usepackage{setspace}

\usepackage{xcolor}
\definecolor{dullmagenta}{rgb}{0.4,0,0.4}   
\definecolor{darkblue}{rgb}{0,0,0.4}
\usepackage[bookmarks,colorlinks,breaklinks]{hyperref}  
\hypersetup{linkcolor=red,citecolor=blue,filecolor=dullmagenta,urlcolor=darkblue} 

\RequirePackage[charter, greekuppercase=italicized,cal=cmcal]{mathdesign}
\usepackage[T1]{fontenc}

\usepackage{verbatim}
\usepackage{xcolor}

\usepackage{braket}
\newcommand{\ketbra}[1]{|#1\rangle\langle #1|}

\renewcommand{\epsilon}{\varepsilon}
\renewcommand{\phi}{\varphi}
\def\eps{\epsilon}
\def\tr{{\rm Tr}}

\def\eps{\epsilon}
\def\id{\mathbbm{1}}

\usepackage{nicefrac}

\usepackage{xspace}

\definecolor{mblue}{rgb}{0.368417, 0.506779, 0.709798}
\definecolor{morange}{rgb}{0.880722, 0.611041, 0.142051}
\definecolor{mgreen}{rgb}{0.560181, 0.691569, 0.194885}
\definecolor{mred}{rgb}{0.922526, 0.385626, 0.209179}
\definecolor{mpurple}{rgb}{0.528488, 0.470624, 0.701351}

\usepackage{tikz}
\usetikzlibrary{arrows}
\usetikzlibrary{calc}
\usepackage{pgfplots}
\usepgfplotslibrary{fillbetween}

\newtheorem{theorem}{Theorem}
\newtheorem{prop}{Proposition}
\newtheorem{lemma}{Lemma}
\newtheorem{corollary}{Corollary}

\usepackage{titling}

\posttitle{\par\end{center}}
\predate{}
\postdate{}
\begin{document}

\title{\large {\bf On privacy amplification, lossy compression, and their duality to channel coding}}
\author{
{\normalsize 
\href{http://orcid.org/0000-0003-2302-8025}{\color{black} Joseph M.\ Renes}}\\
{\small \emph{Institute for Theoretical Physics, ETH Z\"urich}}\\
{\small \emph{Department of Physics, University of Zürich}}
}

\date{\vspace{-\baselineskip}}

\maketitle

\begin{abstract}
We examine the task of privacy amplification from information-theoretic and coding-theoretic points of view.
In the former, we give a one-shot characterization of the optimal rate of privacy amplification against classical adversaries in terms of the optimal type-II error in asymmetric hypothesis testing. 
This formulation can be easily computed to give finite-blocklength bounds and turns out to be equivalent to smooth min-entropy bounds by Renner and Wolf [\href{http://dx.doi.org/10.1007/11593447_11}{Asiacrypt 2005}] and Watanabe and Hayashi [\href{http://dx.doi.org/10.1109/ISIT.2013.6620720}{ISIT 2013}], as well as  
 a bound in terms of the $E_\gamma$ divergence by Yang, Schaefer, and Poor [\href{http://arxiv.org/abs/1706.03866}{arXiv:1706.03866 [cs.IT]}]. 
In the latter, we show that protocols for privacy amplification based on linear codes  can be easily repurposed for channel simulation. 
Combined with known relations between channel simulation and lossy source coding, this implies that privacy amplification can be understood as a basic primitive for both channel simulation and lossy compression. 
Applied to symmetric channels or lossy compression settings, our construction leads to protocols of optimal rate in the asymptotic i.i.d.\ limit.
Finally, appealing to the notion of channel duality recently detailed by us in [\href{http://dx.doi.org/10.1109/TIT.2017.2754921}{IEEE Trans.\ Info.\ Theory {\bf 64}, 577 (2018)}], we show that linear error-correcting codes for symmetric channels with quantum output can be transformed into linear lossy source coding schemes for classical variables arising from the dual channel. 
This explains a ``curious duality'' in these problems for the (self-dual) erasure channel observed by Martinian and Yedidia [Allerton 2003; \href{http://arxiv.org/abs/cs/0408008}{arXiv:cs/0408008}] and partly anticipates recent results on optimal lossy compression by polar and low-density generator matrix codes.
\end{abstract}

\section{Introduction}
Packing and covering are at the core of most simple information processing primitives. 
In noisy channel coding, for instance, where inputs lead to probability distributions over output symbols, one would like to pack as many of these distributions into the space of all possible output distributions such that no two of them overlap significantly. 
This gives an error-correcting code, as the associated inputs can be reliably inferred from the channel output. 
Covering is in some sense dual to packing, as the goal is to find a set of distributions whose empirical average approximates (``covers'') a target distribution. 
In channel simulation, for instance, we would like to approximate the channel output for a given input with the minimal possible amount of additional randomness. 

In this paper we examine the simple covering task of privacy amplification, also known as randomness extraction, both from an information-theoretic as well as a coding-theoretic point of view. 
The goal of privacy amplification, originally introduced in \cite{bennett_privacy_1988}, is to deterministically transform a given random variable $Y$, which may be correlated with $Z$, into the largest possible new random variable $V$ which is uniformly-distributed and independent of $Z$. 
Regarding $Z$ as information held by an adversary or eavesdropper Eve and $Y$ as the variable each held by Alice and Bob, privacy amplification can be understood as a means of extracting a random secret key from information partially correlated with the adversary. 

The natural information-theoretic question is how much randomness can be extracted, and the answer depends on the setting. 
In cryptography, one is interested in making as few assumptions on the correlations to the eavesdropper as possible and usually considers constraints formulated in terms of the min-entropy the adversary has about $Y$, which is related to the maximal probability of guessing $Y$~\cite{impagliazzo_pseudo-random_1989,nisan_randomness_1996}. 
One can also consider adversaries holding quantum information, as opposed to classical, information, but this will not be our focus. 

Instead, we will consider the setting where the complete distribution $P_{YZ}$ is known and $Z$ is classical.
In \S\ref{sec:PA} we give upper and lower bounds on the optimal rate of privacy amplification in a one-shot setting that are formulated in terms of asymmetric hypothesis testing. 
In particular, the minimal type-II error of discriminating between the actual distribution $P_{YZ}$ and an uncorrelated distribution $R_Y{\times} Q_Z$ plays an important role (see Theorem~\ref{thm:PA}), where $R_Y$ is the uniform distribution and $Q_Z$ is arbitrary. 
Previous work in \cite{bennett_privacy_1988,bennett_generalized_1995,renner_simple_2005,hayashi_tight_2013,watanabe_non-asymptotic_2013,hayashi_uniform_2016} is based the smooth min-entropy and its relaxations, though see also \cite{yassaee_non-asymptotic_2013} and the very recent \cite{yang_wiretap_2017}. 
Our converse bound is reminiscient of the metaconverse in channel coding~\cite{nagaoka_strong_2001,polyanskiy_channel_2010-1}\cite[Lemma 4.7]{hayashi_quantum_2017} not only in appearance, but also because it leads to tight, computationally-tractable bounds at finite blocklengths. 
The converse turns out to be equivalent to both the smooth min-entropy bound of Renner and Wolf, Theorem 1 of \cite{renner_simple_2005}, and to the recently formulated $E_\gamma$ bound of Yang, Schaefer, and Poor, Lemma 5 of \cite{yang_wiretap_2017}. 
Moreover, whenever the converse is nontrivial in that the bound on the optimal key size is smaller than the size of the input alphabet, then the converse is also equivalent to the smooth min-entropy bound of Watanabe and Hayashi, Theorem 1 of \cite{watanabe_non-asymptotic_2013}. 
Thus, ultimately we do not need to relax the smooth min-entropy bounds to obtain good finite blocklength bounds. 

Turning to coding theory, in \S\ref{sec:coding} we show that privacy amplification can be used as a primitive to construct protocols for channel simulation and lossy compression. 
First we show that privacy amplification based on linear functions can be used for simulating the action of a given channel $W:X\to Y$ on a given input $X$, such that the simulation error in the latter is precisely equal to to the security parameter in the former. 
The idea behind the construction, stated in detail in Proposition~\ref{prop:pa2sim}, is to consider privacy amplification of the channel \emph{output} $Y$ relative to the \emph{input} $X$. 
If we extend the function from $Y$ to $V$ to be reversible, say $g:Y\to (T,V)$, then only $T$ needs to be transmitted from the encoder to the decoder in order to reconstruct $Y$ by applying $g^{-1}$ to $T$ and common randomness $V$. 
By considering linear functions, we can immediately infer the size of $T$ to be $|Y|/|V|$. 
For symmetric channels, this is sufficient to achieve the optimal rate of communication required for the simulation task, provided the amount of common randomness available to the encoder and decoder is large enough.  

As shown in \cite{steinberg_simulation_1996,winter_compression_2002}, simulating the optimal channel in the rate-distortion function gives a means of turning channel simulation into lossy compression. 
Hence, privacy amplification can also be used to perform lossy compression, the precise details of which are stated in Corollary~\ref{cor:pa2lossy}. 
For sources and distortion functions symmetric in a certain sense, such as the canonical example of compressing a uniformly-random input and considering Hamming distortion, our construction achieves the rate-distortion bound.

Finally, \S\ref{sec:duality} shows that lossy compression can be accomplished by repurposing a good error-correcting code for an appropriate ``dual channel'' as recently investigated by us in \cite{renes_duality_2018}. 
In particular, suppose that $X$ is a random variable to be compressed and reconstructed as $X'$ according to a given distortion measure, and $W=P_{X|X'}$ is the optimal channel in the associated rate-distortion function. 
Then we show that if a code $C$ is good for the dual channel $W^\perp$, then there exists a similarly good lossy compression scheme for $X$, where the reconstructed $X'$ are based on codewords of $C^\perp$ (see Corollary~\ref{cor:duality} for precise details). 
While it happens that the encoder of the channel code and the decompressor are related, the lossy encoder is unrelated to the channel decoder.  
Hence, no guarantees can be made on the efficiency of the lossy compressor even if efficient channel decoding is known to be possible. 

The dual channel usually has a quantum output, but one important exception is the erasure channel.  
In this case, for a channel with erasure probability $q$, the associated lossy compression problem is precisely the binary erasure quantization considered by Martinian and Yedidia~\cite{martinian_iterative_2003}, who established that in this case channel codes can be converted to lossy source codes and \emph{vice versa}.
Thus, we can understand the forward implication as resulting from the deeper structure of duality of codes and channels. 
Establishing this more general relation precisely is one of the main goals of this paper.

\section{Mathematical setup}
We shall only consider random variables with a finite alphabet and will treat their associated probability distributions (probability mass functions) as vectors.
For a random variable $X$ with alphabet $\mathcal X$, we denote the probability mass function as $P_X$ and consider it to be an element of $\mathbb R^{|\mathcal X|}$.
Joint distributions are labelled by all the relevant random variables, and $R$ denotes the uniform distribution. 
Product distributions are denoted by $\times$, e.g.\ $P{\times}Q$, which corresponds to the tensor product at the level of the vector representation.

Events and observables can also be treated as elements of $\mathbb R^{|\mathcal X|}$, and in particular the set of tests will be important for our purposes. 
These are simply vectors whose entries lie in the interval $[0,1]$. 
For a test $\Lambda$ and probability $P_X$, the probability of the test itself will be denoted $\langle \Lambda,P_X\rangle$, which denotes the Euclidean inner product.

We shall also have occasion to consider quantum states and tests, and this notation can also be employed in the quantum setting. 
Probability distributions $P$ are replaced by density operators $\rho$, positive operators of unit trace on $\mathbb C^{|\mathcal X|}$, tests by positive operators on the same space whose eigenvalues do not exceed unity, and the inner product by the Hilbert-Schmidt inner product $\langle \Lambda,\rho\rangle=\tr[\Lambda\rho]$. 

The variational distance of two distributions $P$ and $Q$ is defined by $\delta(P,Q):=\max_{0\leq \Lambda\leq \id}\langle \Lambda,P-Q\rangle$, where $\id$ denotes the vector of all ones. 
From this definition it is immediate that $\delta(P,Q)$ satisfies the triangle and data processing inequalities, and it is easy to show that $\delta(P,Q)=\tfrac12\|P-Q\|_1$. 

Given two distributions, the set of all pairs of probabilities achievable by all possible tests forms the testing region $\mathcal R(P,Q)$ in the unit square. 
We shall make use of the lower boundary of this region, given by the function
\begin{align}
\label{eq:beta}
\beta_\alpha(P,Q):=\min_\Lambda\{\langle\Lambda,Q\rangle:\langle\Lambda,P\rangle\geq \alpha,0\leq\Lambda\leq \id\}\,.
\end{align}
This can be interpreted as the minimal type-II error in an asymmetric hypothesis test between $P$ and $Q$, when the type-I error is constrained to be smaller than $1-\alpha$. 
From the definition it is immediate that $\beta_\alpha$ satisfies the data processing inequality. 
That is, for any stochastic map (channel) $W$, $\beta_\alpha(P,Q)\leq \beta_\alpha(W(P),W(Q))$, since the optimal test $\Lambda^\star$ for $\beta_\alpha(W(P),W(Q))$ induces a feasible test $\Lambda'$ for $\beta_{\alpha}(P,Q)$ by $\langle \Lambda',P\rangle=\langle \Lambda^\star,W(P)\rangle$.

The optimization in \ref{eq:beta} is a linear program (see, e.g.\ \cite{vanderbei_linear_2013}), whose dual formulation is  
\begin{align}
\label{eq:dual}
\beta_\alpha(P,Q)= \max_{\mu,S}\{\mu \alpha-\langle \id, S\rangle:\mu P-S\leq Q,\mu\geq 0,S\geq 0\}\,.
\end{align}
Complementary slackness conditions for the primal and dual programs lead to the well-known Neyman-Pearson lemma~\cite{neyman_problem_1933} that the optimal test $\Lambda$ satisfies $\Lambda(x)=1$ for $\mu P(x)>Q(x)$, $\Lambda(x)=0$ for $\mu P(x)<Q(x)$, and the values of $\Lambda(x)$ for $x$ with $\mu P(x)=Q(x)$ are chosen so that the type-I error is $1-\alpha$. 
Here $\mu$, the optimal value in the dual, is the cutoff (inverse) likelihood ratio for deciding between $P$ and $Q$; clearly the optimal $S$ is just $S(x)=\max\{\mu P(x)-Q(x),0\}$.  
Thus, the region $\mathcal R(P,Q)$ is the convex hull of the points $(\alpha_k,\beta_{\alpha_k})$ obtained by tests of the form $\Lambda(x)=1$ for $x$ such that $P(x)\geq \gamma Q(x)$ for some $\gamma\geq 0$, and zero otherwise. 

The quantity $\beta_\alpha(P,Q)$ is equivalent to the divergence $E_\gamma(P,Q):=P[\frac{P(x)}{Q(x)}\geq \gamma]-\gamma Q[\frac{P(x)}{Q(x)}\geq \gamma]$ in that 
\begin{align}
\label{eq:egammabetaalpha}
E_\gamma(P,Q)=\alpha(\gamma)-\gamma\beta_{\alpha(\gamma)}(P,Q)\,,
\end{align}
where $\alpha(\gamma)=P[\frac{P(x)}{Q(x)}\geq \gamma]$~\cite[Theorem 21]{polyanskiy_channel_2010}. 
To see this, note that \eqref{eq:beta} implies $\beta_{\alpha(\gamma)}(P,Q)\leq Q[\frac{P(x)}{Q(x)}\geq \gamma]$, and therefore $E_\gamma$ is upper-bounded by the righthand side of \eqref{eq:egammabetaalpha}. 
On the other hand, with $\mu=1/\gamma$ in \eqref{eq:dual}, it follows that $\gamma\beta_{\alpha(\gamma)}\geq \alpha(\gamma)-\sum_{x:P(x)\geq \gamma Q(x)} P(x)-\gamma Q(x)=\alpha(\gamma)-E_\gamma(P,Q)$.
Rerunning the argument but leaving the optimization over $S$ in the dual implies  that $E_\gamma(P,Q)=\max_\Lambda\langle \Lambda,P\rangle-\gamma\langle \Lambda,Q\rangle$. 
In this context we also mention the bound 
\begin{align}
\label{eq:betaupper}
\beta_\alpha(P,Q)\leq \frac\alpha\gamma\,,
\end{align} 
which holds for $\gamma$ such that $P[\frac{P(x)}{Q(x)}\geq \gamma]\geq \alpha$ (cf.\ \cite[Equation 2.68]{polyanskiy_channel_2010}). 
To derive it, note that $\Lambda=\id[\frac{P(x)}{Q(x)}\geq \gamma]$ is feasible for $\beta_\alpha(P,Q)$ by assumption, and therefore $\beta_\alpha(P,Q)\leq \langle \Lambda,Q\rangle=\sum_{x:Q(x)\leq P(x)/\gamma} Q(x)\leq \tfrac 1\gamma\sum_{x:Q(x)\leq P(x)/\gamma}P(x)=\alpha/\gamma$.

It is not difficult to see that the variational distance of $P$ and $Q$ is the length of the longest vertical line segment one can place inside above (or below) the diagonal inside the testing region. 
The following proposition gives bounds on $\beta_\alpha(P,Q)$ just in terms of $\alpha$ and $\delta(P,Q)$. 
\begin{lemma}
\label{lem:vdbeta}
For any distributions $P$ and $Q$ and $\alpha\in[0,1]$, 
\begin{align}
 \alpha-\delta(P,Q)\leq \beta_\alpha(P,Q)\leq \alpha\left(1-(1-\alpha)\delta(P,Q)\right)\,.
 \end{align}
\end{lemma}
\begin{proof}
Supposing $\Lambda$ is the optimal test in $\beta_{\alpha}(P,Q)$, it follows immediately that $\delta(P,Q)\geq \langle \Lambda,P-Q\rangle=\alpha-\beta_{\alpha}(P,Q)$.
For the upper bound, let $\Gamma$ be the optimal test in $\delta(P,Q)$ and set $\alpha^\star=\langle \Gamma,P\rangle$. 
By this definition we have $\delta(P,Q)=\langle \Gamma,P-Q\rangle$. 
Next, set $\Lambda=a \id+b\Gamma$ for $a$ and $b$ to be determined later, but such that $\langle \Lambda,P\rangle=\alpha$, i.e.\ $a+b\alpha^\star=\alpha$. 
It then follows that $\beta_\alpha(P,Q)\leq \langle \Lambda,Q\rangle=\alpha-b\delta(P,Q)$.
Now there are two cases to consider, $\alpha$ smaller or larger than $\alpha^\star$.
For the former, the choice $a=0$ and $b=\alpha^\star/\alpha$ ensures that $\Lambda$ is a valid test. 
Then $\beta_\alpha(P,Q)\leq \alpha(1-\tfrac1{\alpha^\star}\delta(P,Q))$, which implies the desired bound since $1/\alpha^\star\geq 1\geq 1-\alpha$.
For the latter, choosing $b=1-a$ and $a=(\alpha-\alpha^\star)/(1-\alpha^\star)$ again leads to a valid test. 
Here we have $\beta_\alpha(P,Q)\leq \alpha(1-\tfrac{1-\alpha}{1-\alpha^\star}\delta(P,Q))$ which implies the desired bound since $\frac1{1-\alpha^\star}\geq 1\geq \alpha$. 
\end{proof}
Note that $\alpha^\star=P[P(x)\geq Q(x)]$ appearing in the proof is precisely the value of $\alpha$ for which the vertical distance between the diagonal and the lower boundary of $\mathcal R(P,Q)$ is the variational distance. 

The lower bound corresponds to the upper bound in \cite[Proposition 3.2]{jensen_generalized_2013}, while the upper bound is due to Fr\'ed\'eric Dupuis. 
We mention in passing that the lower bound claimed in Proposition 3.2 is in error. 
In particular, for $P$ and $Q$ deterministic and completely uniform distributions of a binary-valued random variable, respectively, we have $\delta(P,Q)=\frac 12$ and $\beta_\alpha(P,Q)=\frac\alpha2$. 
This violates the ostensible bound for $\alpha<\nicefrac12$. 
This example also shows the claimed Pinsker-like inequality $\frac{1-\alpha}{\alpha}\delta(P,Q)\leq -\log \frac1\alpha\beta_\alpha(P,Q)$ is violated for $\alpha<\nicefrac 13$. 

For a joint distribution $P_{YZ}$, the min-entropy of $Y$ conditioned on $Z$ is related to the largest conditional probability $P_{Y|Z=z}(y)$, maximized over possible values of $Z$:\footnote{Note that another oft-used definition of the min-entropy is based on averaging over $P_Z$.} 
\begin{align}
\label{eq:minentropy}
H_{\min}(Y|Z)_P:=-\log \max_{y\in \mathcal Y,z\in \text{supp}(P_Z)}P_{Y|Z=z}(y)\,.
\end{align}
Following \cite{watanabe_non-asymptotic_2013}, its smoothed version is based on replacing the joint distribution, but not the marginal $P_Z$, with a nearby distribution $Q_{YZ}$ that decreases the ratio:\footnote{Note that the smoothing in \cite{renner_simple_2005} is different: $Q_{YZ}$ is required to be smaller than $P_{YZ}$ but have a normalization not less than $1-\eps$.}
\begin{align}
H_{\min}^\eps(Y|Z)_P:=-\log \min_{Q:\delta(P,Q)\leq \eps}\max_{y\in \mathcal Y,z\in \text{supp}(P_Z)}\frac{Q_{YZ}(y,z)}{P_Z(z)}\,.
\end{align}
Using the dual form of $\delta(P,Q)$ as $\min\{\langle \id, T\rangle:T\geq P-Q,T\geq 0\}$, the smoothed min-entropy can be expressed as a linear program:
\begin{align}
2^{-H_{\min}^\eps(Y|Z)_P}:=
\min\{\lambda:\lambda \id_YP_Z\geq Q_{YZ},T_{YZ}\geq P_{YZ}-Q_{YZ},\langle \id_{YZ},T_{YZ}\rangle\leq \eps,\langle \id_{YZ},Q_{YZ}\rangle=1;\lambda,T,Q\geq 0\}\,.
\end{align}

\section{Bounds on extractable randomness}
\label{sec:PA}
Given a joint distribution $P_{YZ}$, the task of randomness extraction, or privacy amplification, of $Y$ relative to $Z$ is to apply a function $f:Y\to V$ such that the resulting distribution $P_{VZ}$ is essentially the same as $R_V{\times} P_Z$. 
This setup is depicted in Figure~\ref{fig:PA}. 
We sometimes refer to $f$ as the extractor or extractor function, though note that in the cryptography community an extractor refers to a set of functions useful for generating randomness from a source characterized only in terms of min-entropy. 
We measure closeness by the variational distance, and say that $f$ is a protocol for $(k,\eps)$ privacy amplification for $P_{YZ}$ when $\delta(P_{VZ},R_V{\times}P_Z)=\eps$ and $\log |V|=k$.

\begin{figure}[h]
\centering
\includegraphics{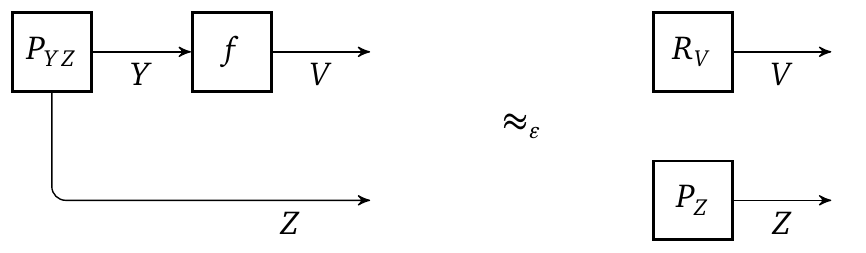}
\caption{\label{fig:PA}Schematic representation of privacy amplification (randomness extraction) of $Y$ relative to $Z$. 
The function $f$ should produce a random variable $V$ which is $\eps$-close to being uniformly random and independent of $Z$, as measured by the variational distance. 
}
\end{figure}

Letting $K_\eps(Y|Z)_P$ be the largest $K\in \mathbb N$ such that there exists a $(\log K,\eps)$ privacy amplification protocol for $P_{YZ}$, we can show the following result. 
\begin{theorem}
\label{thm:PA}
For any joint distribution $P_{YZ}$, 
\begin{align}
K_\eps(Y|Z)_P&\leq \min_{\eta\in [0,1-\eps]}\tfrac1\eta \beta_{\eps+\eta}(P_{YZ},\id_Y{\times} P_Z)\qquad\text{and}\label{eq:converse}\\
K_\eps(Y|Z)_P&\geq \max_{\eta\in[0,\eps]}\Big\lfloor\tfrac{4\eta^2}{\eps-\eta}\max_{Q_Z}\beta_{\eps-\eta}(P_{YZ},\id_Y{\times}\, Q_Z)\Big\rfloor\,.\label{eq:forward}
\end{align} 
\end{theorem}

We will show the second inequality, the achievability statement, by employing two-universal hashing~\cite{carter_universal_1979}. Thus, linear functions are capable of achieving the stated bounds. 
The argument is a combination of the argument given for channel resolvability by Hayashi for channels with classical~\cite[Lemma 2]{hayashi_general_2006} and quantum~\cite[Lemma 9.2]{hayashi_quantum_2017} output with that of the leftover hashing lemma of~\cite{tomamichel_leftover_2011}, adapted to yield a bound involving $\beta_\alpha$.
The first inequality, the converse, is an adaptation of the converse involving min-entropy reported in~\cite[\S III.A]{tomamichel_hierarchy_2013}, itself based on \cite[\S 8.2.2]{tomamichel_framework_2012}. 
Both the underlying achievability and converse arguments also apply when the adversary holds quantum information (i.e.\ $Z$ is quantum), but this is decidedly not the case for the hypothesis testing versions. 
We shall remark on the steps that fail for quantum side information.

\subsection{Proof of the converse}
Let us begin by showing the converse.
It relies on the following lemma, the hypothesis-testing version of the statement that min-entropy of $Y$ conditioned on $Z$  cannot increase by the application of a function to $Y$. 
\begin{lemma}
\label{lem:function}
For any function $f:Y\to V$, $\beta_\alpha(P_{VZ},\id_V{\times} P_Z)\leq \beta_\alpha(P_{YZ},\id_Y{\times} P_Z)$ for all $\alpha\in [0,1]$. 
\end{lemma}
\begin{proof}
Using the conditional distributions $P_{Y|Z=z}$ we can construct a stochastic map $W$ from $P_{VZ}$ back to $P_{YZ}$:
\begin{align}
W(y,z|v,z')=\delta_{z,z'}\frac{\delta_{f(y),v}P_{Y|Z=z}(y)}{\sum_{y'}\delta_{f(y'),v}P_{Y|Z=z}(y')}\,.
\end{align}
Clearly $W$ is stochastic and $W(P_{VZ})=P_{YZ}$. 
By the data processing inequality we have $\beta_\alpha(P_{VZ},\id_V{\times} P_Z)\leq \beta_\alpha(P_{YZ},W(\id_V{\times} P_Z))$. 
Now observe that 
\begin{subequations}
\begin{align}
[W(\id_V{\times} P_Z)](y,z)
&=\sum_{v,z'}W(y,z|v,z')P_Z(z')\\
&=P_Z(z)\frac{P_{Y|Z=z}(y)}{\sum_{y'}\delta_{f(y),f(y')}P_{Y|Z=z}(y')}\\
&\leq P_Z(z)\,.
\end{align}
\end{subequations}
Hence $\beta_\alpha(P_{YZ},W(\id_V{\times} P_Z))\leq \beta_\alpha(P_{YZ},\id_Y{\times} P_Z)$, completing the proof.
\end{proof}

Now suppose $f:Y\to V$ is the extractor function of a $(\log K_\eps(Y|Z)_P,\eps)$ privacy amplification protocol. 
By Lemma~\ref{lem:vdbeta}, we have $\beta_{\eps+\eta}(P_{VZ},R_V{\times} P_Z)\geq \eta$ for any $0\leq \eta\leq 1-\eps$.
This is equivalent to $\beta_{\eps+\eta}(P_{VZ},\id_V{\times} P_Z)\geq \eta|V|=\eta K_\eps(Y|Z)_P$, whence
Lemma~\ref{lem:function} and a minimization over $\eta$ gives \eqref{eq:converse}.\\

\newenvironment{psmallmatrix}
  {\left(\begin{smallmatrix}}
  {\end{smallmatrix}\right)}

Lemma~\ref{lem:function} does not hold for quantum $Z$, as one can find counterexamples. 
That this might be the case can be anticipated by noticing that in the proof we make use of the distribution of $Y$ conditioned on the value of $Z$, for which there is no quantum analog. 
To state a specific counterexample, consider the qubit density operators $\phi_j=\ketbra{j}$ for $j=0,1$ and $\{\ket{j}\}$ an orthonormal basis, along with $\phi_2=\ketbra{+}$, where $\ket{+}=\tfrac1{\sqrt 2}(\ket{0}+\ket 1)$. 
Suppose $Y$ is uniformly-distributed in $\{0,1,2\}$ and take the state of $Z$ given $Y=y$ to be $\phi_y$; the marginal state of $Z$ averaged over $Y$ is $\bar\phi=\tfrac13(\id+\ketbra{+})$
For the function $f$ mapping 0 and 1 to 1 and 2 to itself, then $V$ has distribution $(\nicefrac 23,\nicefrac13)$ and corresponding conditional states $\tfrac12\id$ and $\varphi_2$. 
In the lower bound of $\beta_\alpha(\rho_{VZ},\id_V{\times}\rho_Z)$ we can choose $\mu=2$ and $S=\frac13\ketbra{{-}}$ so that $\beta_\alpha(\rho_{VZ},\id_V{\times}\rho_Z)\geq 2\alpha-\tfrac13$. 
In the upper bound of  $\beta_\alpha(\rho_{YZ},\id_Y{\times}\rho_Z)$ we can take the test for $Y=y$ to be a scaled projection onto the zero eigenspace of $\tfrac49\phi_y-\bar\phi$ for $Y=0,1$ and zero for $Y=2$. 
In particular, $\Lambda_0=\frac \alpha{6}\begin{psmallmatrix}9 &-3\\ -3& 1\end{psmallmatrix}$ and $\Lambda_1=\frac \alpha{6}\begin{psmallmatrix}1 &-3\\ -3& 9\end{psmallmatrix}$, and $\Lambda_3=0$. 
Then $\langle \Lambda,\rho_{YZ}\rangle=\alpha$ and $\langle \Lambda,\id_Y\otimes \rho_Z\rangle=\frac43\alpha$, meaning $\beta_\alpha(\rho_{YZ},\id_Y{\times}\rho_Z)\leq \frac43\alpha$. 
For $\alpha>\nicefrac12$ this is smaller than $2\alpha-\tfrac13$, so this example shows that the lemma cannot hold for arbitrary $\alpha$. 
It is an open question if it holds for $\alpha<\nicefrac12$. 

The example also shows that the converse bound itself does not hold for quantum $Z$. 
The function $f$ results in an output quantum state $\rho_{VZ}$ with $\eps=\nicefrac16$, meaning $K_{\nicefrac16}(Y|Z)_\rho\geq 2$. 
But if we choose $\eta=\nicefrac12$ so that $\alpha=\eps+\eta=\nicefrac 23$, then from the converse we have $K_{\nicefrac 16}(Y|Z)_\rho\leq \frac1{\eta}\frac 43\alpha=\frac{16} 9$, a contradiction.

\subsection{Proof of achievability}
Now we move to the proof of the direct part.
Let $f:Y\to V$ be an arbitrary function with $|V|=K_\eps(Y|Z)_P$, and set $\Delta:=\delta(P_{VZ},R_V\times P_Z)$. 
For arbitrary test $\Lambda_{YZ}$ on $YZ$, define the rescaled probability distributions $\hat P_{VZ}$ and $\bar P_{VZ}$ such that $\hat P_{VZ}+\bar P_{VZ}=P_{VZ}$ and $\hat P_{VZ}(v,z)=\sum_{y:f(y)=v}P_{YZ}(y,z)\Lambda_{YZ}(y,z)$. 
Using the triangle inequality we have
\begin{subequations}
\begin{align}
\Delta
&\leq \delta(P_{VZ},R_V{\times} \bar P_Z)+\delta(R_V{\times} \bar P_Z,R_V{\times} P_Z)\\
&= \delta(\hat P_{VZ}+\bar P_{VZ},R_V{\times} \bar P_Z)+\delta(\bar P_Z, P_Z)\\
&\leq \delta(\bar P_{VZ},R_V{\times} \bar P_Z)+\delta(\bar P_Z, P_Z)+\tfrac12\|\hat P_{VZ}\|_1
\end{align}
\end{subequations}
Due to the form of $\bar P_{VZ}$ and $\hat P_{VZ}$, the latter two terms combine to give $\langle \Lambda_{YZ},P_{YZ}\rangle$, so that
\begin{align}
\label{eq:firstcamp}
\Delta\leq \tfrac12\|\bar P_{VZ}-R_V{\times} \bar P_Z\|_1+\langle \Lambda_{YZ},P_{YZ}\rangle\,.
\end{align}

To bound the first term in this expression, we use the fact that $(\sum_{i} |x_i|)^2\leq \sum_i s_i^{-1}x_i^2$ for any set of $s_i>0$ such that $\sum_i s_i=1$. 
This follows by writing $\|x\|_1=\|\frac x{\sqrt s}\sqrt s\|_1$ and applying the H\"older inequality with $p=q=2$. 
Choosing $s_{vz}=Q_Z(z)/|V|$ for some normalized distribution $Q_Z$ with strictly positive probabilities to be determined later, we have
\begin{subequations}
\label{eq:tdbound}
\begin{align}
\|\bar P_{VZ}-R_V{\times} \bar P_Z\|_1^2
&\leq |V| \sum_{vz}Q_Z(z)^{-1}(\bar P_{VZ}(v,z)-\tfrac1{|V|}\bar P_Z(z))^2\\
&=|V|\sum_{vz} Q_Z(z)^{-1}\bar P_{VZ}(v,z)^2-\sum_z Q_Z(z)^{-1}\bar P_Z(z)^2\,.
\end{align}
\end{subequations}
Now let us deal with the summation over $v$ in the first term. 
With $\Lambda'(y,z)=1-\Lambda(y,z)$ and omitting the $YZ$ random variable subscripts, we have
\begin{subequations}
\label{eq:commuteex}
\begin{align}
\sum_v \bar P_{VZ}(v,z)^2
&=\sum_v \Big(\sum_{y:f(y)=v}P(y,z)\Lambda'(y,z)\Big)\Big(\sum_{y':f(y')=v}P(y,z)\Lambda'(y,z)\Big)\\
&=\sum_y P(y,z)^2\Lambda'(y,z)^2+\sum_{y\neq y'}\delta_{f(y)=f(y')}P(y,z)P(y',z)\Lambda'(y,z)\Lambda'(y',z)\,.
\end{align}
\end{subequations}
Taking the expectation over $f$ chosen uniformly at random from a family of two-universal hash functions, we then obtain
\begin{subequations}
\begin{align}
\mathbb E_f \sum_v \bar P_{VZ}(v,z)^2 
&=\sum_y P(y,z)^2\Lambda'(y,z)^2+\tfrac1{|V|}\sum_{y\neq y'}P(y,z)P(y',z)\Lambda'(y,z)\Lambda'(y',z)\\
&\leq \sum_y P(y,z)^2\Lambda'(y,z)^2+\tfrac1{|V|}\bar P_Z(z)^2\,.
\end{align}
\end{subequations}
Using this in \eqref{eq:tdbound} gives
\begin{align}
\label{eq:almostthere}
\mathbb E_f \|\bar P_{VZ}-R_Z{\times} \bar P_Z\|_1^2
&\leq |V|\sum_{yz} Q_Z(z)^{-1}P_{YZ}(y,z)^2\Lambda'_{YZ}(y,z)^2\,.
\end{align}

Finally, let $\Lambda_{YZ}$ be the optimal test in $\beta_\eta(P_{YZ},\id_Y\times Q_Z)$, so that $\langle \Lambda_{YZ},P_{YZ}\rangle=\eta$ and $\beta_\eta(P_{YZ},\id_Y\times Q_Z)\leq \mu \eta$ for the optimal $\mu$ in the dual formulation \eqref{eq:dual}. 
By the properties of the optimal test it follows that 
 $\Lambda'_{YZ}(y,z)(\mu P_{YZ}(y,z)-Q_Z(z))\leq 0$. 
Therefore, $\Lambda'_{YZ}(y,z)P_{YZ}(y,z)Q_Z(z)^{-1}\leq \tfrac1\mu\Lambda'_{YZ}(y,z)$.
From \eqref{eq:almostthere} we then obtain
\begin{subequations}
\begin{align}
\mathbb E_f \|\bar P_{VZ}-R_V{\times} \bar P_Z\|_1^2
&\leq\frac{|V|}\mu \sum_{yz} P_{YZ}(y,z)\Lambda'_{YZ}(y,z)^2\\
&\leq \frac{|V|}\mu\\
&\leq \frac{\eta |V|}{\beta_\eta(P_{XY},\id_X\times Q_Y)}\,.
\end{align}
\end{subequations}
By Jensen's inequality, $\mathbb E_f \|\bar P_{VZ}-R_V{\times} \bar P_Z\|_1^2\geq (\mathbb E_f \|\bar P_{VZ}-R_V{\times} \bar P_Z\|_1)^2$.
Returning to \eqref{eq:firstcamp}, we have
\begin{align}
\label{eq:land}
\mathbb E_f \Delta \leq \eta+\frac12\sqrt{\frac{\eta |V|}{\beta_\eta(P_{YZ},\id_Y\times Q_Z)}}\,.
\end{align}
Choosing $\eta=\eps-\delta$ and $|V|=\lfloor\frac{4\delta^2}{\eps-\delta}\max_{Q_Z}\beta_{\eps-\delta}(P_{YZ},\id_Y\times Q_Z)\rfloor$ ensures there exists an $f$ such that $\Delta\leq \eps$.
Since $\beta_\alpha(P,Q)$ is continuous in $Q$, we may maximize over all $Q_Z$, not just those with strictly positive probability. This completes the proof.

Were $Z$ quantum rather than classical, the test $\Lambda_{YZ}$ would take the form of a set of positive operators $\Lambda_y$, and we would immediately be confronted with the possibility that $\Lambda_y$ may not commute with the conditional quantum states of $Z$ given $Y=y$, i.e.\ the quantum versions of $P_{Z|Y=y}$. 
In the proof we assume these objects commute, e.g.\ in \eqref{eq:commuteex}. 
One method of dealing with this issue, as done in \cite{hayashi_quantum_2017}, is to ``pinch'' the quantum states (remove the off-diagonal elements) to restore commutation, and appeal to bounds between the pinched and unpinched states. 
It is unclear if this can be done in combination with the steps taken here to end up with a bound in terms of $\beta_\alpha$.

\subsection{Comparison of the bounds}

After publication of the initial version of this manuscript, Wei Yang pointed out that the bounds contained herein are related to those recently derived in \cite{yang_wiretap_2017}. 
The achievability bound can be derived from \cite[Lemma 2]{yang_wiretap_2017}, which states that there exists an $(K,\eps)$ privacy amplification scheme such that, for all distributions $Q_Z$ and $\gamma>0$, 
\begin{align}
\eps \leq E_\gamma(P_{YZ},R_Y{\times}Q_Z)+\sqrt{\frac{\gamma K}{4|Y|}\mathbb E_{P_{YZ}}[\exp(-|i(Y{:}Z)-\log \gamma|)]}\,,
\end{align}
where $i(y,z)=\log P_{YZ}(y,z)-\log R_Y(y)Q_Z(z)$. 
To get back to \eqref{eq:forward}, observe that the expectation term is necessarily smaller than 1 and $E_\gamma\leq P_{YZ}[\frac{P_{YZ}(y,z)}{R_Y(y)Q_Z(z)}\geq \gamma]$.
Now set $\alpha=P_{YZ}[\frac{P_{YZ}(y,z)}{R_Y(y)Q_Z(z)}\geq \gamma]$ and use \eqref{eq:betaupper} to get $\eps\leq \alpha+\frac12\sqrt{\alpha K/\beta_\alpha(P_{YZ},\id_Y{\times}Q_Z)}$, which is \eqref{eq:land}.

Lemma 5 of \cite{yang_wiretap_2017} is the converse statement that every $(K,\eps)$ privacy amplification protocol for $P_{YZ}$ ($Z$ classical) must satisfy
\begin{align}
\label{eq:Egammaconverse}
\eps\geq E_{|Y|/K}(P_{YZ},R_Y{\times} P_Z)\,.
\end{align}
\begin{prop}
The converse bounds \eqref{eq:converse} and \eqref{eq:Egammaconverse} are equivalent.
\end{prop}
\begin{proof}
To obtain \eqref{eq:converse} from this expression, use the variational form of $E_\gamma$ to write $\eps\geq \max_{\Lambda} \langle \Lambda,P_{YZ}\rangle-\frac{|Y|}{K}\langle \Lambda,R_Y{\times}P_Z\rangle$. 
Now consider tests that yield the vertices of $\mathcal R(P_{YZ},R_Y{\times}P_Z)$, specifically $\Lambda_k=\id[\frac{P(x)}{Q(x)}\geq \gamma_k]$ such that $\alpha_k=\langle \Lambda_k,P_{YZ}\rangle$ is larger than $\eps$. 
Since $\beta_{\alpha_k}(P_{YZ},R_Y{\times} P_Z)= \langle \Lambda_k,R_Y{\times}P_Z\rangle$, we have 
\begin{align}
\label{eq:epsboundbeta}
\eps\geq \alpha_k-\frac{|Y|}{K}\beta_{\alpha_k}(P_{YZ},R_Y{\times} P_Z)\,.
\end{align}
The function $\alpha\mapsto \beta_\alpha$ interpolates linearly between vertices, and therefore the relation holds for arbitrary $\alpha$. 
Setting $\alpha-\eps=\eta$ and optimizing over $\eta$ recovers \eqref{eq:converse}. 

To show that \eqref{eq:converse} is equivalent, we use \eqref{eq:Egammaconverse} to choose a suitable $\eta$ in the optimization. 
Specifically, suppose $K^\star$ is the optimizer in \eqref{eq:Egammaconverse} so that $\eps=E_{|Y|/K^\star}(P_{YZ},R_Y{\times} P_Z)$, and let $\gamma^\star=|Y|/K^\star$. 
Now define $\eta=\alpha-\eps$ for $\alpha=P_{YZ}[\frac{P_{YZ}(y,z)}{R_Y(y)P_Z(z)}\geq \gamma^\star]$. 
Clearly $\eta\leq 1-\eps$ since $\alpha$ must be positive. 
On the other hand, $\eta$ must be positive since we have $\eps=\alpha-\gamma^\star (R_Y{\times} P_Z)[\frac{P_{YZ}(y,z)}{R_Y(y)P_Z(z)}\geq \gamma^\star]$, the second term of which is positive. 
Indeed, by \eqref{eq:egammabetaalpha}, we must have $\eps=\eps+\eta-\gamma^\star \beta_{\eps+\eta}(P_{YZ},R_Y{\times} P_Z)$.
Rearranging the expression gives $|Y|\beta_{\eps+\eta}(P_{YZ},\id_Y{\times} P_Z)=\eta K^\star$, and therefore $K_\eps(Y|Z)_P\leq K^\star$. 
Since we have $K_\eps(Y|Z)_P\leq K^\star$ from the previous argument, this implies that the choice of $\eta$ is optimal. 
\end{proof}
The equivalence also holds when $Z$ is quantum, because the lower boundary of the testing region is still given by likelihood ratio tests as used above. 
However, the above counterexample to Lemma 

We can express the converse bound using a quantity similar to the smooth min-entropy by appealing to Proposition 13.6 of \cite{liu_Egamma_2017}. 
For convenience, define $\lambda_{\min}^\eps(Y|Z)_P:=2^{-H_{\min}^\eps(Y|Z)_P}$ and let $\bar \lambda_{\min}^\eps(Y|Z)_P$ be the same optimization without the normalization condition. 
Note that dropping the normalization condition means that we do not explicitly require $Q_{YZ}$ as a variable at all, and we can instead simply define  
\begin{align}
\bar \lambda_{\min}^\eps(Y|Z)_P:=\min\{\lambda:\lambda \id_YP_Z\geq P_{YZ}-T_{YZ},\langle \id_{YZ},T_{YZ}\rangle\leq \eps;\lambda,T_{YZ}\geq 0\}\,.
\end{align}
This holds because for any feasible $\lambda$ and $T_{YZ}$ we can pick any $Q_{YZ}\geq 0$ satisfying $P_{YZ}-T_{YZ}\leq Q_{YZ}\leq \lambda \id_YP_Z$.
In fact, since we can replace any feasible $T_{YZ}$ with one satisfying $P_{YZ}-T_{YZ}\geq 0$ without affecting feasibility of $\lambda$, $P_{YZ}-T_{YZ}$ can be assumed to be positive without loss of generality. 
\begin{prop}
The converse bound \eqref{eq:Egammaconverse} is equivalent to 
\begin{align}
\label{eq:WHconverserelax}
\frac1{K_\eps(Y|Z)_P}\geq \bar\lambda_{\min}^\eps(Y|Z)_P\,.
\end{align}
\end{prop}
\begin{proof}
First note that for any $\gamma$ and possibly non-normalized $Q_{YZ}$ we have 
\begin{subequations}
\label{eq:prop6}
\begin{align}
E_{\gamma}(P_{YZ},R_Y{\times}P_Z)
&=\max_\Lambda \,\langle \Lambda, P_{YZ}\rangle -\gamma\langle \Lambda,R_Y{\times}P_Z\rangle\\
&=\max_\Lambda \,\langle \Lambda, P_{YZ}\rangle -\langle \Lambda, Q_{YZ}\rangle +\langle\Lambda,Q_{YZ}\rangle -\gamma\langle \Lambda,R_Y{\times}P_Z\rangle\\
&\leq \delta(P_{YZ},Q_{YZ})+E_\gamma(Q_{YZ},R_Y{\times}P_Z)\,.
\end{align}
\end{subequations}
(This is \cite[Proposition 6]{liu_Egamma_2017}.) 
Note that here we have extended the domain of $\delta$ to include non-normalized arguments, and now the function is no longer symmetric in its arguments. 
Nevertheless, its dual formulation is still $\delta$ is $\delta(P,Q)=\min\{\tr[T]:T\geq P-Q,T\geq 0\}$.
Now set $Q_{YZ}=P_{YZ}-T_{YZ}$ for $T_{YZ}$ an optimizer in $\bar \lambda_{\min}^\eps(Y|Z)_P$. 
Then $\delta(P_{YZ},Q_{YZ})\leq \eps$, while the second term in \eqref{eq:prop6} is zero for $\gamma\geq |Y|\bar \lambda_{\min}^\eps(Y|Z)_P$.
Therefore, setting $\gamma=|Y|/K_\eps(Y|Z)_P$ and using \eqref{eq:WHconverserelax} gives \eqref{eq:Egammaconverse}.  

For the other direction, pick any $\gamma>0$ and define $Q_{YZ}$ as the pointwise minimum of $P_{YZ}$ and $\gamma R_Y{\times}P_Z$. 
The optimal test $\Lambda$ in $\delta(P_{YZ},Q_{YZ})$ is just the indicator onto the positive part of the difference $P_{YZ}-Q_{YZ}$, which by construction is equal to the indicator onto the positive part of $P_{YZ}-\gamma R_Y{\times}P_Z$.  
Then we have $\delta(P_{YZ},Q_{YZ})=E_\gamma(P_{YZ},R_Y{\times}P_Z)$. 
Therefore, since $Q_{YZ}\leq \lambda\id_YP_Z$ for $\lambda=\frac{\gamma}{|Y|}$ and $\tr[T_{YZ}]\leq E_\gamma(P_{YZ},R_Y{\times}P_Z)$ for $T_{YZ}$ the positive part of $P_{YZ}-\gamma R_Y{\times}P_Z$, it follows that $\bar \lambda_{\min}^\eps(Y|Z)_P\leq \gamma/|Y|$ for $\eps=E_\gamma(P_{YZ},R_Y{\times}P_Z)$.
Setting $\gamma=|Y|/K_\eps(Y|Z)_P$ and using \eqref{eq:Egammaconverse} implies \eqref{eq:WHconverserelax}.
\end{proof}

The $Q_{YZ}$ appearing in both directions of the above proof is smaller than $P_{YZ}$, which is precisely the kind of smoothing employed by Renner and Wolf~\cite{renner_simple_2005}. 
In particular, the converse in their Theorem 1 is the statement
\begin{align}
\label{eq:RWconverse}
\frac1{K_\eps(Y|Z)_P}\geq \widehat\lambda_{\min}^\eps(Y|Z)_P\,,
\end{align}
where $\widehat \lambda_{\min}^\eps(Y|Z)_P:=\min\{\lambda:\lambda \id_YP_Z\geq Q_{YZ},Q_{YZ}\leq P_{YZ},\langle \id_{YZ},Q_{YZ}\rangle\geq 1-\eps;\lambda,Q\geq 0\}$.
In light of the above, it is perhaps not too suprising that 
\begin{align}
\widehat\lambda_{\min}^\eps(Y|Z)_P=\bar\lambda_{\min}^\eps(Y|Z)_P\,,
\end{align}
and thus \eqref{eq:RWconverse} is equivalent to \eqref{eq:WHconverserelax}.
This follows because, on the one hand, $Q_{YZ}=P_{YZ}-T_{YZ}$ for $T_{YZ}$ optimal in the latter is feasible in the former, and on the other, $T_{YZ}=P_{YZ}-Q_{YZ}$ for $Q_{YZ}$ optimal in the former is feasible for the latter. 

Meanwhile, the smooth min-entropy bound of Watanabe and Hayashi \cite[Theorem 1]{watanabe_non-asymptotic_2013} is simply
\begin{align}
\label{eq:WHconverse}
\frac1{K_\eps(Y|Z)_P}\geq \lambda_{\min}^\eps(Y|Z)_P\,.
\end{align}
(This bound can be shown using the stochastic map $W$ as in \eqref{eq:converse}, which is essentially the same as their proof.)
Nominally, then, the equivalent bounds \eqref{eq:converse}, \cite[Lemma 5]{yang_wiretap_2017}, \eqref{eq:WHconverserelax}, and \cite[Theorem 1]{renner_simple_2005} are relaxations \eqref{eq:WHconverse}. 
However, these are all equivalent whenever the former are nontrivial.
\begin{prop}
If $\bar \lambda_{\min}^\eps(Y|Z)_P>\frac1{|Y|}$, then $\bar\lambda_{\min}^\eps(Y|Z)_P= \lambda_{\min}^\eps(Y|Z)_P$.
\end{prop}
\begin{proof}
Since the optimal $T_{YZ}$ is positive and has no entry larger than that of $P_{YZ}$, the smallest normalization of the possible $Q_{YZ}$ is less than one.
The largest is  $\lambda |Y|$, and hence a normalized $Q_{YZ}$ can be found whenever the optimal $\lambda$ is larger than $\frac1{|Y|}$. 
\end{proof}


\begin{figure}[h]
\centering
\begin{tikzpicture}
\begin{axis}[
enlargelimits = false,
  ymin=0.19,
  ymax=0.43,
  xlabel={blocklength $n$},
  ylabel={extraction rate \quad$\tfrac1n \log\, K_\eps(Y^n|Z^n)$},
  scaled x ticks = false,
  xmin=100,
  xmax=5000,
  xtick={1000,2000,3000,4000,5000},
  xticklabels={1000,2000,3000,4000,5000},
  legend style={font=\scriptsize,at={(.985,0.02)},anchor=south east},
  legend cell align=left,
]
\addplot[name path=old,mred,semithick] table {oldupper.dat};
\addlegendentry{upper bound \cite[(10)]{watanabe_non-asymptotic_2013}}

\addplot[name path=new,morange,semithick] table {betterupper.dat};
\addlegendentry{upper bound \eqref{eq:converse}, \cite[(154)]{yang_wiretap_2017},}

\addlegendimage{empty legend}
\addlegendentry{\phantom{upper bound} \cite[Th.\ 1]{renner_simple_2005}, \cite[Th.\ 1]{watanabe_non-asymptotic_2013}}

\addplot[domain=100:10000,samples=200] {0.499916 - 6.00364/sqrt(x)};
\addlegendentry{normal approximation\ \cite[Th.\ 3]{watanabe_non-asymptotic_2013}}

\addplot[mgreen,semithick] table {betterlower.dat};
\addlegendentry{lower bound \cite[(150)]{yang_wiretap_2017}}

\addplot[mblue,semithick] table {hybridlower.dat};
\addlegendentry{lower bound \cite[(12)]{watanabe_non-asymptotic_2013}}

\addplot[mpurple,semithick] table {hypolower.dat};
\addlegendentry{lower bound \eqref{eq:forward}, \cite[(9)]{watanabe_non-asymptotic_2013}}


\end{axis}
\end{tikzpicture}
\caption{\label{fig:plot}Comparison of finite blocklength bounds on randomness extraction from $Y^n$ relative to $Z^n$ for the i.i.d.\ case of $Z$ obtained from uniform $Y$ through a BSC of crossover probability $0.11$ and a target security parameter of $\eps=10^{-10}$. The asymptotic rate for this example is $\nicefrac 12$.
}
\end{figure}
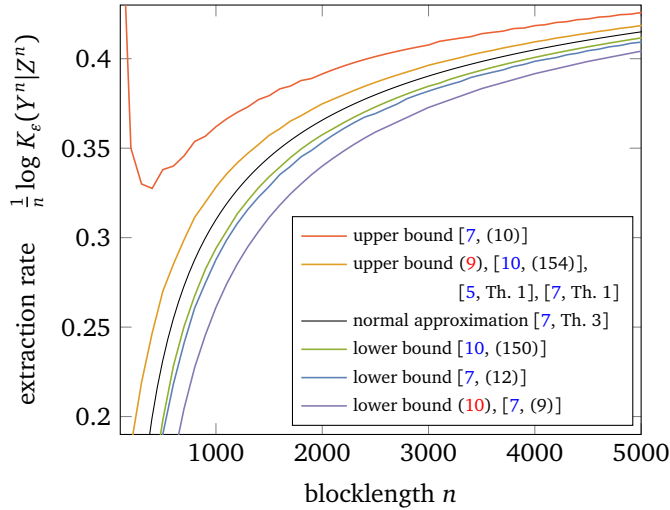

To compare the bounds for a fixed example, consider $Z$ obtained from uniform $Y$ through a binary symmetric channel with crossover probability $0.11$ and a target security parameter of $\eps=10^{-10}$. 
We follow \cite[Theorem 35]{polyanskiy_channel_2010-1} in computing $\beta_\alpha(P_{YZ}^{\times n},\id_{Y^n}{\times}R_{Z^n})$. 
Assuming $Q_Z$ uniform, the lower bound of \eqref{eq:forward} and that of \cite[Equation 9]{watanabe_non-asymptotic_2013} are essentially identical, with \eqref{eq:forward} improving on the latter by only four bits. 
The hybrid bound of \cite[Theorem 6]{watanabe_non-asymptotic_2013} is considerably better. 
Better yet is Theorem 18 in \cite{yang_wiretap_2017}\footnote{Equations 150 and 153 contain a small error: The term $g_n(\gamma)$ inside the square root should be $g_n(\gamma)^2$ (Wei Yang, private communication).}.
On the converse side, \eqref{eq:converse} and \cite[Lemma 5]{yang_wiretap_2017} allow us to sidestep the numerical difficulties associated with smoothing and still compute the min-entropy bound \cite[Theorem 1]{watanabe_non-asymptotic_2013}. 
This yields a substantial improvement over the information-spectrum relaxation \cite[Equation 10]{watanabe_non-asymptotic_2013}. 
The relaxation of \cite[Lemma 29]{hayashi_uniform_2016} is substiantially better, only two bits worse in this range.   
(All the bounds mentioned here are formulated specifically for classical $Z$.) 
It is satisfying to see the normal approximation obtained by disregarding the $O(\log n)$ term in \cite[Theorem 3]{watanabe_non-asymptotic_2013} lies midway between the tightest upper and lower bounds. 
Note that in this example the $O(\log n)$ term is provably absent from the asymptotic expansion, as shown in \cite[Theorem 19]{yang_wiretap_2017}.

\section{Repurposing randomness extraction}
\label{sec:coding}

\subsection{Channel simulation}
\label{sec:channelsim}
Given a channel $W:X\to Y$, the task of channel simulation is to reproduce the joint input and output statistics of $W$ applied to $X$ by making use of an ideal (noiseless) channel and common randomness at the encoder and decoder. 
Randomness is of course necessary to simulate the stochastic nature of $W$, and the main information-theoretic question is to characterize the amounts of communication and randomness required for a given channel $W$ and input $X$. 
A stronger variant is \emph{universal} channel simulation, in which the simulation works for any input $X$, i.e.\ the encoder is not constructed using the knowledge of the particular input $X$; in this paper we are only interested in the non-universal case.

Let us first specify the setup more concretely. 
Given an input distribution $P_X$, the channel $W$ leads to a joint distribution $P_{XY}$. 
An $(m,r,\eps)$ protocol for simulation of $W$ applied to $X$ consists of an encoding map $E:(X,V)\to T$ and a decoding map $D:(T,V)\to Y'$, such that $\log |T|=m$, $\log |V|=r$, and  $\delta(P_{XY},P_{XY'})\leq \eps$, where $Y'=D(E(X,V),V)$ and $V$ is a uniformly-distributed random variable.  

\begin{figure}[h]
\centering
\includegraphics{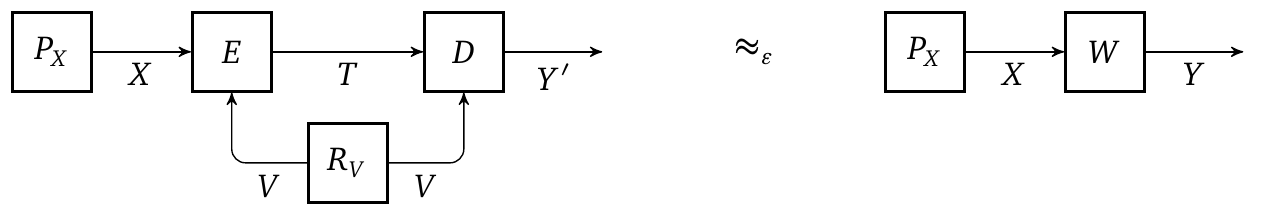}
\caption{\label{fig:channelsim} Simulation of the channel $W$ acting on input random variable $X$ by means of an ideal channel and common randomness.}
\end{figure}

The following shows that privacy amplification can be repurposed for channel simulation (see also \cite{renes_theory_2013}). 
\begin{prop}
\label{prop:pa2sim}
Let $W:X\to Y$ be an arbitrary channel and $P_X$ an arbitrary input distribution. 
Suppose the linear function $f$ is a protocol for $(k,\eps)$ privacy amplification of $Y$ relative to $X$. 
Then a $(\log |Y|-k,k,\eps)$ protocol for simulating the action of $W$ on $X$ can be constructed from $f$ and $P_{XY}$. 
\end{prop}
\begin{proof}
The linear function $f$ can always be extended to a reversible linear function $g:Y\to (T,V)$, and given the joint distribution $P_{XY}$, $g$ induces a conditional distribution $P_{T|XV}$.
Interpreted as a channel, $P_{T|XV}$ defines the encoder $E$.
It requires $\log |V|$ bits of random input and produces a message of $\log |T|=\log |Y|/|V|$ bits. 
Meanwhile, the decoder $D$ is just  $g^{-1}$, calling the output $Y'$ instead of $Y$.

By design, $E(V,X)=T$, meaning the combined action of $D$ and $E$ on $P_{XV}$ gives back $P_{XY}$. 
In the protocol, this combined action is instead applied to $P_X{\times}\, R_{V}$, producing $P_{XY'}$. 
By the data processing inequality we therefore have $\delta(P_{XY},P_{XY'})\leq \delta(P_{XV},P_X{\times}\, R_{V})$. 
But the premise is that the latter quantity is bounded from above by $\eps$, so the proof is complete.  
\end{proof} 

Since we are considering linear $f$, it makes sense to associate the function to a linear code. 
Suppose we consider $f$ to output the syndromes of some code $C$ encoding $n-k$ bits into $n$, i.e.\ $f(y^n)=H y^n$, where $H$ is the $k\times n$ parity check matrix of $C$. 
Then, if $g$ is to be reversible, the output in $T$ must correspond to the message encoded in $C$, and the action of $g^{-1}$ is therefore to produce a codeword of $C$ determined by $t$, offset to a coset determined by $v$. 
To see this more concretely, extend $H$ to an invertible $n\times n$ matrix $M$ which defines $g$, say $M=\begin{psmallmatrix}\bar M\\ \hat M\end{psmallmatrix}$ with $\hat M=H$. 
The reconstruction operation will apply $M^{-1}$, so that $y^n=M^{-1}(t\oplus v)$. 
Now let $M'=(M^{-1})^T$ and set $M'=\begin{psmallmatrix}\bar M'\\ \hat M'\end{psmallmatrix}$ for $\bar M'$ an $(n-k)\times n$ matrix.
In terms of $M'$, the action of $g^{-1}$ is given by $(y^n)^T=(t^T\oplus v^T)M'=t^T \bar M'\oplus v^T \hat M'$. 
Since $MM'^T=\id$, $H\bar M'^T=0$, meaning $\bar M'$ is the generator matrix of $C$, and the action of $g^{-1}$ is as claimed.

For channels $W$ whose input $X$ induces a uniformly-random output $Y$, the above construction directly leads to protocols which achieve the optimal communication rate of channel simulation in the asymptotic i.i.d.\ scenario, assuming unlimited common randomness. 
The amount of communication required is $|T^n|=|Y^n|/K_\eps(Y^n|X^n)$, and so in the i.i.d.\ case with uniform $Y$ we have $\lim_{n\to \infty}\tfrac1n\log |T^n|=\frac1n(\log |Y^n|-H(Y^n|X^n))=I(X{:}Y)$ while $\eps\to 0$ (cf.\ \cite[Theorem 2]{bennett_entanglement-assisted_2002}).  
To achieve the mutual information rate in the general case of a non-uniform $Y$, one option is to concatenate the above protocol with data compression of $Y^n$, though we will not pursue this further here (for a similar approach applied to noisy channel communication, see \cite{renes_noisy_2011}).

\subsection{Lossy source coding}

The task of lossy source coding is to compress a random variable $X$ to a smaller alphabet such that the reconstructed $X'$ is close to $X$ as measured by a distortion function $d:X\times X'\to \mathbb R_+$. 
There are two main approaches to specifying the distortion constraint, by requiring either a fixed average value or a fixed probability of exceeding a specified target value.
In the latter case one defines a $(k,d^\star,\eps)$ protocol for lossy compression of $X\sim P_X$ to consist of an encoder $E:X\to C$ and a decoder $D:C\to X'$ such that $P[d(X,D\circ E(X))>d^\star]\leq \eps$ and $\log |C|=k$. 
This excess distortion probability can be expressed as the expectation under $P_{XX'}$ of a test function $\Lambda$ which is 1 whenever $d(x,x')>d^\star$ and zero otherwise, i.e.\ $P[d(X,D\circ E(X))>d^\star]=\langle \Lambda,P_{XX'}\rangle$.
This setup is depicted in Figure~\ref{fig:lossy}. 
In the former case, a $(k,\bar d)$ protocol consists of an encoder and decoder such that $\mathbb E_{P_{XX'}} d(X,X')=\bar d$. 
The latter is a stronger criterion, since any $\eps$-good scheme with target value $d^\star$ can be converted into an average scheme with average distortion less than $d^\star+\eps \max_{x,x'}d(x,x')$ (see, e.g.\ \cite[Lemma 6]{datta_one-shot_2013-1}). 

\begin{figure}
\centering
\includegraphics{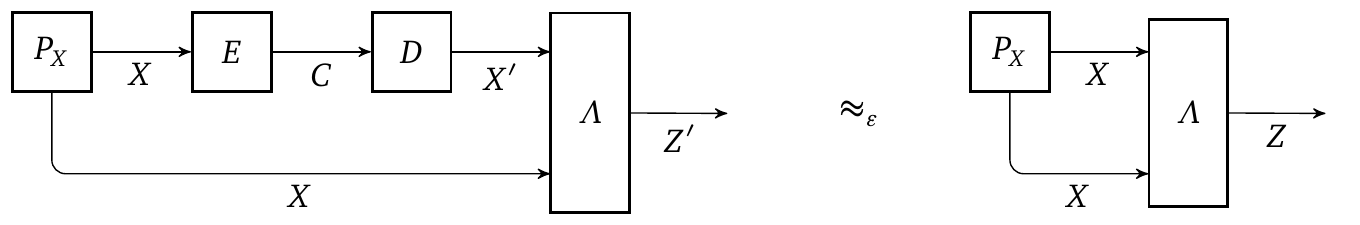}
\caption{\label{fig:lossy} Lossy compression of the random variable $X$. $\Lambda$ is a binary-valued test that indicates $1$ if $d(x,x')\geq d^\star$ and zero otherwise.}
\end{figure}

Channel simulation is one way to construct a lossy compression protocol, as shown in \cite{steinberg_simulation_1996,winter_compression_2002} (and later in the quantum setting in \cite{luo_channel_2009,datta_quantum_2013,datta_one-shot_2013-1}). 
For the case of excess distortion probability, we formalize the statement as follows. 
\begin{prop}
Given $P_X$, a distortion function $d$ and a target distortion $d^\star$, suppose that $W:X\to X'$ is a channel such that the target distortion for $X$ and $X'$ is exceeded with probability no larger than $\eps'$. 
Then a $(k,\eps)$ protocol for channel simulation of $W$ applied to $X$ can be used to construct a $(k,d^\star,\eps+\eps')$ lossy compression protocol for $X$. 
\end{prop}
\begin{proof}
Let $X''$ be the output of the channel simulation protocol, for which $\delta(P_{XX'},P_{XX''})\leq \eps$. 
By the properties of the variational distance, $\delta(P_{XX''},P_{XX'})\geq \langle \Lambda,P_{XX''}\rangle-\langle \Lambda,P_{XX'}\rangle$, and therefore the probability $\langle \Lambda,P_{XX''}\rangle$ that $X$ and $X''$ exceed the target distortion is no larger than $\eps+\eps'$. 
\end{proof}

As opposed to channel simulation, in lossy compression there is no need for randomized encoding or decoding operations. 
Simulation protocols adapted to lossy compression can in principle be derandomized: Since the output distribution $P_{XX''}$ is a mixture over the values of $V$ and the distortion probability is linear in $P_{XX''}$, one may as well pick that value of $V$ leading to the smallest excess distortion probability.  
This still leaves the randomness in the mapping $P_{T|X=x,V=v}$, but this can be derandomized in the same manner. 

In the i.i.d.\ setting it is perhaps more natural to consider the case of average distortion directly, since here we are interested in symbolwise distortion functions of the form $d(x^n,x'^n)=\tfrac1n \sum_{i}d_{\text{sym}}(x_i,x_i')$ for some function $d_{\text{sym}}$, as well as constructing protocols by simulating single symbol channels $W:X_i\to X'_i$ (rather than $W:X^n\to X'^n$). 
In this case we have
\begin{prop}
\label{prop:symboldistort}
Given $P_X^{\times n}$ and a symbolwise distortion function $d$ with $0\leq d_{\text{sym}}(x_i,x_i')\leq 1$ for all $x_i,x'_i$, suppose that $W:X\to X'$ is a channel such that $\mathbb E_{P_X} d_{\text{sym}}(X,X')=\bar d\leq 1$. 
Then a $(k,\eps)$ simulation protocol for $W^{\times n}$ acting on $X^n$ can be used to construct a $(k,\bar d+\eps)$ lossy compression scheme for $X^n$.
\end{prop} 
\begin{proof}
Let $X''^n$ be the output of the simulation protocol. 
Since the protocol is $\eps$-good, $\delta(P_{X_iX'_i},P_{X_iX''_i})\leq \eps$ for all $i\in [n]$. 
But by the definition of the variational distance, 
\begin{subequations}
\begin{align}
\delta(P_{X_iX'_i},P_{X_iX''_i})
&=\max_{0\leq f\leq 1}\sum_{x,y}f(x,y)(P_{X_iX''_i}(x,y)-P_{X_iX_i'}(x,y))\\
&\geq \sum_{x,y}d_{\text{sym}}(x,y)(P_{X_iX''_i}(x,y)-P_{X_iX_i'}(x,y))\\
&=\mathbb E_{P_{X_iX''_i}} d_{\text{sym}}(X_i,X''_i)-\bar d\,.
\end{align}
\end{subequations}
Thus $\mathbb E_{P_{X_iX''_i}} d_{\text{sym}}(X_i,X''_i)\leq \bar d+\eps$, and averaging over $i$ gives the desired result. 
\end{proof}

In the i.i.d.\ setting the optimal rate is given by the rate-distortion function $R(\bar d)=\inf_{W:d(X,W(X))\leq \bar d}I(X:W(X))$. 
When the optimal channel $W$ in this expression gives a uniformly-random $X'$, then we can employ an optimal rate channel simulation protocol to construct an optimal rate lossy compression procedure. 
Hence, in this case we can ultimately rely on privacy amplification of $X'$ relative to $X$ by linear functions to perform lossy compression of $X$ at the optimal rate. 
For fixed blocklength we have the following corollary:
\begin{corollary}
\label{cor:pa2lossy}
In the setting of Proposition~\ref{prop:symboldistort}, suppose $P_{XX'}$ is a joint distribution such that $\mathbb E_{P_{XX'}}d_{\text{sym}}(X,X')\leq \bar d$. 
Then any linear $(k,\eps)$ privacy amplification protocol for $X'^n$ relative to $X^n$ can be used to construct a $(n-k,\bar d+\eps)$ lossy compression scheme for $X^n$.
\end{corollary}

The standard example of a uniformly-random source $X$ and Hamming distortion $d_{\text{sym}}(x,x')=\delta_{x,x'}$ has a uniform $X'$, since the optimal channel is just a binary symmetric channel with crossover probability $\bar d$~\cite[\S10.3.1]{cover_elements_2006}.
So, too, does the binary erasure quantization example of Martinian and Yedidia~\cite{martinian_iterative_2003}. 
Here the input $X$ has alphabet $\{0,1,?\}$, with probabilities $(1-e)/2$, $(1-e)/2$ and $e$ for some $0\leq e\leq 1$, respectively, and the symbol distortion function $d_{\text{sym}}(x,x')$ is 0 if $x=?$ or $x=x'$, and 1 otherwise. 
As reported by Kostina and Verd\'u in \cite[Equation 202]{kostina_fixed-length_2012}, the rate distortion function for this case is $R(\bar d)=(1-e)(1-h_2(\tfrac{\bar d}{1-e}))$ (after adapting the notation to the present setting).
The optimal channel is a concatenation of the map which randomly assigns $?$ inputs to 0 and 1, but leaves those input values untouched, and a binary symmetric channel with crossover probability $\bar d/(1-e)$.
It therefore has uniform output over \{0,1\} for any value of $e$.

\section{Lossy compression from channel coding}
\label{sec:duality}

By making use of duality relations for channels and codes, we can show that linear error-correcting codes can be used to build lossy source codes. 
Suppose $X$ is the random variable to be lossily compressed, and $P_{XX'}$ is the optimal joint distribution in the rate distortion function. 
This induces a channel $W=P_{X|X'}$ given the marginal $P_{X'}$; note that this channel is defined in the opposite sense to \S\ref{sec:channelsim}. 
Corollary~\ref{cor:pa2lossy} establishes that lossy compression can be constructed from privacy amplification of the input of $W$ relative to its output. 
But, following \cite{renes_duality_2011,renes_duality_2018}, this task is dual to channel coding (or lossless compression) for the dual channel $W^\perp:X'\to B$, where now $W^\perp$ is a channel whose output is quantum-mechanical (see \cite[\S3.2]{renes_duality_2018} for a precise definition). 
Here we restrict attention to symmetric channels $W$. 

Specifically, Corollary 8 of \cite{renes_duality_2018} ensures that a $(k,\eps)$ code $C$ for $W^\perp$ (where $\eps$ is the average error probability under the optimal decoder) leads to a linear $(k,2\sqrt{\eps})$ privacy amplification protocol for $X'$ relative to $X$. 
The extractor function $f$ is given by the generator matrix $G$ of $C$ acting to the right, i.e.\ the parity check matrix of $C^\perp$. 
In fact, using Theorem 5.1 of \cite{renes_duality_2011} we can improve the security parameter to $\sqrt{2\eps}$ (the difference stems from the use of the max-entropy in Corollary 8, which involves an optimization over the marginal of $Z$ rather than using actual marginal $P_Z$ directly). 
Combining this with Corollary~\ref{cor:pa2lossy}, we obtain
\begin{corollary}
\label{cor:duality}
For $P_{X|X'}$ the optimal conditional distribution appearing in the rate-distortion function, let $W^\perp$ be the dual channel according to \cite{renes_duality_2018}. Then a $(k,\eps)$ code $C$ for $W^\perp$ can be used to construct an $(n-k,\sqrt{2\eps})$ lossy compression scheme for $X^n$. 
\end{corollary}
The reconstruction operation outputs codewords of $C^\perp$, shifted to a coset determined by the common randomness $V$. 
Meanwhile, the quantizer or compressor is stochastic, based on the conditional distribution  distribution $P_{T|XV}$ as in Proposition~\ref{prop:pa2sim}. 

The dual channels for the examples mentioned above can be explicitly given. 
For the case of Hamming distortion, the optimal channel from $X'$ to $X$ is also a BSC with crossover probability $\delta=\bar d$, which means the dual channel takes the classical input $z$ to the pure state $\ket{\theta_z}=\sqrt{\delta}\ket{0}+(-1)^z\sqrt{1-\delta}\ket 1$. 
For binary erasure quantization, the optimal channel is a concatenation of a BSC with crossover probability $\delta=\bar d/(1-e)$ with an erasure channel with erasure probability $e$. The dual of this channel is computed in Example 3.9 of \cite{renes_alignment_2016}. Its output consists of two independent parts, one classical and one quantum-mechanical. 
The classical part is the just the output of the erasure channel with erasure probability $1-e$, while the quantum part is exactly the output of the dual of the BSC.
Thus, with probability $e$ the input to the channel shows up unchanged in the classical part of the output, but even when the classical part is useless, the quantum part contains some information about the input. 

Note that when $\bar d=0$ in the latter case, the two quantum states $\ket{\theta_z}$ are identical, so the quantum part of the dual output is useless. 
Then the dual is effectively just the erasure channel with erasure probability $1-e$.  
The above relation between error-correcting code $C$ for $W^\perp$ and the use of the dual code $C^\perp$ for privacy amplification of $W$ partly explains the ``curious duality between erased/known symbols in source coding and known/erased symbols in channel coding'' observed by Martinian and Yedidia~\cite{martinian_iterative_2003}. 
By more direct analysis of error-correction and lossy compression for the erasure channel, they show an equivalence between the two tasks. 
Here we have made use of the more general theory of duality and shown one direction of the equivalence, namely that error-correction implies lossy compression. 

This also makes sense of recent results on the optimality of polar codes and spatially-coupled low-density generator matrix codes for lossy compression, shown in \cite{korada_polar_2010-1} and \cite{aref_approaching_2015}, respectively. 
Polar codes are their own duals in the sense that the dual of a polar code is given by using the frozen bits instead of the information bits, which is again a code constructed by the properties of synthesized channels. 
Thus, the fact that polar codes achieve the capacity for classical-quantum channels described above, which follows from the general result of~\cite{wilde_polar_2013}, implies that polar codes achieve the rate-distortion bound for the associated sources.
One simply has to base the scheme on the synthesized inputs with high entropy, exactly as done in \cite{korada_polar_2010-1}.  
Similarly, it is not unreasonable to suspect that spatially-coupled low-density parity check (LDPC) codes achieve the capacity of these classical-quantum channels under optimal decoding. 
This would imply that their duals, low-density generator matrix (LDGM) codes, achieve the rate-distortion bound for the associated sources, precisely as shown in \cite{aref_approaching_2015}. 
Duality also helps explain that LDPC codes are themselves not useful for lossy compression, as observed in \cite{martinian_iterative_2003}, as otherwise their duals, LDGM codes, would be good for channel coding, and this is known not to be the case~\cite{mackay_good_1999}.  
Note that the implications based on duality say nothing about complexity of encoding or decoding operations for either channel coding or lossy compression; indeed, establishing low complexity bounds is the better part of the results of \cite{korada_polar_2010-1} and \cite{aref_approaching_2015}.

\section{Discussion}
We have given new bounds on the optimal rate of privacy amplification in a one-shot setting and seen that the converse bound is equivalent to bounds based on the smooth min-entropy as well as the $E_\gamma$ divergence, but avoids the computational difficulties of smoothing in the finite-blocklength setting. 
While the achievability bound is not the tightest known in the literature, the formulation of both in terms of hypothesis testing has advantages of its own. 
One is the clear connection of information theory to statistics. 
More conceptually, using a common quantity in optimal rate bounds allow us to see the concrete relationship of covering and packing problems more plainly. 
Namely, for approximation parameter $\eps$, rate bounds for covering problems involve $\beta_\eps$, while those for packing problems involve $\beta_{1-\eps}$. 
This supports the notion that packing is dual to covering.
An open question is whether the hypothesis-testing approach can be extended to covering problems involving quantum information. 


We have also shown that privacy amplification is a primitive for constructing channel simulation and lossy compression protocols. 
Doing so enables us to extend the known duality of codes for packing (lossless compression) and covering (privacy amplification) to the covering problem of lossy compression. 
Specifically, coding duality implies that duals of good channel codes lead to good lossy source codes, at least for symmetric channels and lossy compression setups.  
An immediate open question in this context is whether one can go in the other direction, from lossy source coding back to channel coding. 
This was observed to be the case for the binary erasure channel and binary erasure quantization in \cite{martinian_iterative_2003}. 
Perhaps the most straightforward approach to demonstrating this would be to show that lossy source codes can be used for privacy amplification, since duality ensures that a good privacy amplification protocol implies the existence of a good channel code for the dual (cf \cite[Corollary 8]{renes_duality_2018}, \cite[Theorem 5.2]{renes_duality_2011}).
One could also investigate whether duality leads to improved finite blocklength bounds on privacy amplification. 
However, since the dual setup involves a quantum output, this seems doubtful as the bounds available in this case are not as tight as for classical output (see \cite{tomamichel_hierarchy_2013}). 

A much bigger and more tantalizing open question is whether duality also provides a link between algorithms for channel decoding and source quantization, e.g.\ using belief propagation (BP). 
Such a link was shown in Theorem 3 of \cite{martinian_iterative_2003} for the binary erasure coding and quantization problems, and progress on the general case could help in finding new bounds on the performance of BP. 
A first step in this direction would be to extend the notion of BP decoding to the duals of classical channels. 
This is presumably possible by extending our construction in \cite{renes_belief_2017}, which dealt with the dual of the BSC, since any symmetric binary-input classical channel can be regarded as a mixture of BSCs.
Though duality was not used in the construction therein, in retrospect its role is evident. 
In particular, the unitaries for combining quantum information at check and variable nodes could have been determined by appealing to the convolution rules for the BSC and Theorem 1 of \cite{renes_duality_2018}, which states that the dual of a check convolution is the variable convolution of the duals and similarly for the dual of a variable convolution.\footnote{Despite my best efforts, an error in the check node convolution in an initial draft of \cite{renes_duality_2018} persists in the final published version (though not the most recent arXiv version): The check node convolution should read $U_\boxasterisk=\textsc{cnot}_{1\to 2}$. Thanks to Narayanan Rengaswamy for pointing this out, twice!} 
All of which hints concretely to the possibility that duality can shed light on BP, but the details remain to be seen.

\vspace{2mm}
{\bfseries{Acknowledgments.}} I thank Renato Renner and Marco Tomamichel for helpful  discussions. 
Particular thanks to Fr\'ed\'eric Dupuis for providing the upper bound of Lemma~\ref{lem:vdbeta}, Henry Pfister for pointing me to \cite{martinian_iterative_2003}, and Wei Yang for pointing out the relationship between Theorem~\ref{thm:PA} and the results of \cite{yang_wiretap_2017} and \cite{liu_Egamma_2017}. 
This work was supported by the Swiss National Science Foundation (SNSF) via the National Centre of Competence in Research “QSIT”, as well as the Air Force Office of Scientific Research (AFOSR) via grant FA9550-16-1-0245.

\printbibliography[heading=bibintoc,title=References]

\end{document}